\documentclass[reprint,pre,showpacs,amsmath,amssymb,superscriptaddress]{revtex4-1}

\usepackage{graphicx}   

\newtheorem{theorem}{Theorem}[section]

\newenvironment{proof}[1][Proof]{\begin{trivlist}
\item[\hskip \labelsep {\bfseries #1}]}{\end{trivlist}}

\newcommand{\qed}{\nobreak \ifvmode \relax \else
      \ifdim\lastskip<1.5em \hskip-\lastskip
      \hskip1.5em plus0em minus0.5em \fi \nobreak
      \vrule height0.75em width0.5em depth0.25em\fi}

\newcommand{\field}[1]{\mathbb{#1}}

\newcommand{\TYNDALL}{Tyndall National Institute, University College Cork, Lee Maltings, Cork, Ireland}
\newcommand{\MATHUCC}{School of Mathematical Sciences, University College Cork, Cork, Ireland}
\newcommand{\PHYSUCC}{Department of Physics, University College Cork, Cork, Ireland}

\begin{document}

\title{Multi-stabilities and symmetry-broken one-colour and two-colour states\\in closely coupled single-mode lasers}

\author{Eoin Clerkin}
\email{eoin@clerkin.biz}
\affiliation{\TYNDALL}
\affiliation{\PHYSUCC}

\author{Stephen O'Brien}
\affiliation{\TYNDALL}

\author{Andreas Amann}
\affiliation{\TYNDALL}
\affiliation{\MATHUCC}

\date{\today}


\begin{abstract}
We theoretically investigate the dynamics of two mutually coupled identical single-mode semi-conductor lasers.  For small separation and large coupling between the lasers, symmetry-broken one-colour states are shown to be stable.  In this case the light output of the lasers have significantly different intensities while at the same time the lasers are locked to a single common frequency.  For intermediate coupling we observe stable symmetry-broken two-colour states, where both lasers lase simultaneously at two optical frequencies which are separated by up to 150~GHz.  Using a five dimensional model we identify the bifurcation structure which is responsible for the appearance of symmetric and symmetry-broken one-colour and two-colour states. Several of these states give rise to multi-stabilities and therefore allow for the design of all-optical memory elements on the basis of two coupled single-mode lasers. The switching performance of selected designs of optical memory elements is studied numerically.
\end{abstract}

\maketitle

\section{Introduction}

A system of two mutually coupled semiconductor lasers is a simple example of interacting nonlinear oscillators. This system has been experimentally realised using for example edge emitting \cite{HEI01}, quantum dot \cite{HEG07}, VCSEL \cite{FUJ03} and DBR \cite{VAU09} lasers and the observed dynamical phenomena include leader-laggard dynamics \cite{HEI01}, frequency locking and intensity pulsations \cite{WUE05} and bubbling \cite{FLU09,TIA12}.  For recent reviews on the rich dynamical features of coupled semiconductor lasers we refer to \cite{LUE12,SOR13}.  

In the current paper we study this problem theoretically on the basis of a well established rate equation model.  In general, there are two time scales which govern the character of the dynamics in coupled semiconductor lasers, namely the period of the relaxation oscillation $T_R=1/\nu_R$  and the delay time $\tau$ due to the separation between the two lasers.  In the long delay case $\tau\gg T_R$, it is well known that the synchronous state of two coupled lasers is in general not stable \cite{WHI02}, which gives rise to  leader-laggard dynamics and noise induced switching between the  asymmetric states \cite{MUL04}.  In view of applications in small-scale and high-speed devices, we are however interested in the opposite limit, where the delay time is comparable to or even much smaller than the relaxation period.  In this limit the rate equation approach is justified if the distance between the lasers is significantly larger than the optical wavelength itself.  At even shorter distances, composite cavity 
models have recently been used to describe effects due to evanescent lateral coupling \cite{ERZ08,BLA12}.  Using typical values for the relaxation oscillation in the GHz regime and an optical wavelength of around 1~$\mu$m, we therefore now focus on the case of two coupled lasers at distances between 10~$\mu$m and 500~$\mu$m. 

We observe that for small spatial separation and strong optical coupling, the two lasers can mutually lock into one of two stable one-colour states. These states spontaneously break the original symmetry of exchanging the two lasers, and both lasers emit light at precisely the same frequency but at different intensities. In this bi-stable regime the current state of the system can be conveniently observed optically from the amplitude of the light output of the lasers.  There also exists a region of hysteresis, where the symmetric one-colour state is jointly stable with the two symmetry broken one-colour states, thus giving rise to a parameter region of tri-stability. For more moderate coupling strength, the symmetry broken one-colour states become unstable, and instead symmetry broken two-colour states appear, which are similar to the ones observed numerically in \cite{ROG03}.  In this case, both lasers emit light at the same two optical frequencies, however the intensities of the respective colours are 
not identical in both lasers.  

In order to study these and similar transitions in the dynamics  systematically, we introduce a reduced five dimensional model, which allows for a bifurcation analysis using the continuation software AUTO \cite{DOE06}.  The transition from a one-colour to a two-colour symmetry broken state then corresponds to a Hopf bifurcation from a symmetry-broken fixed-point state to a symmetry-broken limit cycle state in the language of bifurcation theory.  We are in particular interested in the fundamental bifurcations, which bound the domain of symmetry broken two-colour states. We identify the relevant codimension two points, which organise the bifurcation scenario in this region, and explain the mechanism which gives rise to the large region of bi-stability between symmetry-broken two-colour states. 

From the technological side there exists currently a strong interest in the development of small-scale devices which are capable of all-optical signal processing.  One particular challenge, which has attracted a significant amount of recent research  activities,  is the design of all-optical memory elements  \cite{HIL04,OSB09,LIU10,CHE11,HEI11,NOZ12,PER12}.  The goal is to design fast and efficient memory elements which can be switched between at least two different stable states via an external optical signal.  At the same time it is also desired that the  state of the memory element is accessible optically, and the optical output from one memory element should be able to trigger the switching in further elements, with little or no intermediate processing.  Based on these criteria, the question arises, if all-optical memory elements can be realised using two mutually interacting identical single-mode lasers.  After having identified a number of interesting regions of multi-stabilities involving one-colour 
and two-colour states, we numerically demonstrate that these states can indeed be exploited for the design of all-optical memory units. By injecting suitable pulses of light into the coupled lasers we are able to switch between different multi-stable states.  In order to assess the technologically relevant speed of the switching events, we define the write time as the minimal injected pulse duration required to trigger the switching event, and the read time as the minimum time after which the state of the memory can be obtained optically.  We find read and write times of less than 100ps, which suggests the possibility of fast and simple all-optical memory elements on the basis of identical closely coupled single-mode lasers.

\section{Rate Equation Model}

\begin{figure}
\includegraphics[width=1.0\columnwidth]{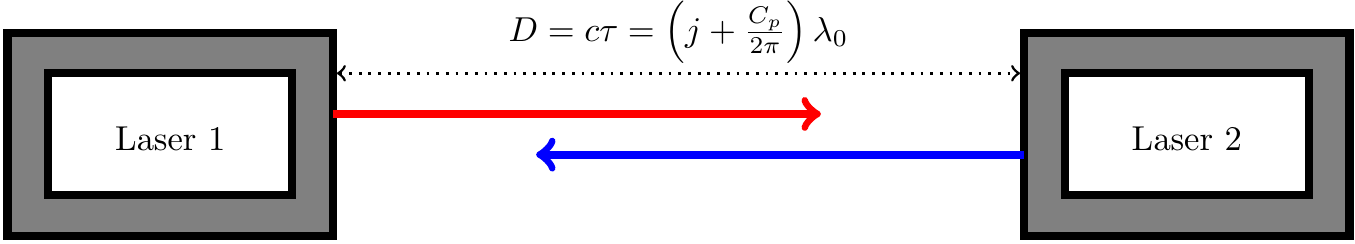}
\caption{\label{fg:scheme} Schematic diagram of two coupled lasers of wavelength $\lambda_0$ separated by a distance $D$, where  $C_p \in [0,2\pi)$ and $j \in \field{N}$.}
\end{figure}

Two identical single-mode semiconductor lasers with free-running wavelength $\lambda_0$ are placed in a face-to-face alignment with a separation $D$ as sketched in Fig.~\ref{fg:scheme}. They are mutually coupled via a certain amount of the light of each entering the other cavity after a delay $\tau = D/c$, where $c$ denotes the speed of light. This scenario has been successfully studied in the literature on the basis of rate equation models \cite{HOH97,MUL02,ROG03,YAN04,ERZ06} and we will use the following system of delay differential equations:

\begin{subequations}\label{eq:mLK}
\begin{eqnarray}
\dot{E_1}(t) =& (1 + i \, \alpha) \, N_1(t) \, E_1(t) + \kappa e^{-i \, {C_p}} E_2 (t -\tau) \label{eq:E1dot}\\
\dot{E_2}(t) =& (1 + i \, \alpha) \, N_2(t) \, E_2(t) + \kappa e^{-i \, {C_p}} E_1 (t -\tau) \label{eq:E2dot} \\
\dot{N_1}(t) =& \frac{1}{T}\left[P - N_1(t) - \left( 1 + 2N_1(t) \right) \, |E_1(t)|^2\right]  \label{eq:N1dot}\\
\dot{N_2}(t) =& \frac{1}{T}\left[P - N_2(t) - \left( 1 + 2N_2(t) \right) \, |E_2(t)|^2\right] \label{eq:N2dot}
\end{eqnarray}
\end{subequations}

The parameters include $P=0.23$, the pumping of electron-hole pairs into each laser; $\alpha=2.6$, the line-width enhancement factor; $T=392$, the ratio of the photonic and carrier lifetimes; $\tau$ the dimensionless delay time between the lasers. The coupling strength $\kappa$ and the coupling phase $C_p\,=\,2\pi ( D\,mod\,\lambda_0)$ are the two main bifurcation parameters in subsequent sections. The model~\eqref{eq:mLK} is dimensionless with time measured in units of the photon lifetime $\tau_p = 1.0204$ ps.  The dynamical variables $N_1$ and $N_2$ denote the population inversions, and $E_1$ and $E_2$ are the slowly varying complex optical fields in laser 1 and laser 2.  The rapidly oscillating physical fields can be recovered via $\tilde{E}_{1,2} = E_{1,2}(t) e^{i \omega_0 t}$, where the optical angular frequency is given by $ \omega_0= 2\pi c / \lambda_0$. The phase factor $e^{-i C_p}$  in system~\eqref{eq:mLK} is therefore a consequence of expressing the time delayed physical fields using slowly 
varying fields via $\tilde{E}_{1,2}(t-\tau) e^{-i \omega_0 t} = e^{-i C_p} E_{1,2}(t-\tau)$.

Any solution of equations~\eqref{eq:mLK} can be multiplied by a common phase factor in both optical fields leading to a $S^1$ symmetry. In addition, as the two lasers are identical, a $\field{Z}_2$ symmetry exists due to the ability to swap the lasers. Mathematically  these two phase space symmetries can be formulated as \cite{ERZ06},
\begin{equation}
\label{eq:symm}
\begin{aligned}
 & \left(E_1,E_2\right) \rightarrow \left(e^{i\,b}E_1,e^{i\,b}E_2\right), \: b \in [0,2\pi) \:  & S^1 \text{ symmetry} \\
 & \left(E_1,E_2,N_1,N_2\right) \rightarrow \left(E_2,E_1,N_2,N_1\right) \: & \field{Z}_2 \text{ symmetry}
\end{aligned}
\end{equation}
Both the $S^1$ and $\field{Z}_2$ symmetries are frequently used and referred to in subsequent sections. 

\section{One-Colour States}\label{sec:CLM}

\begin{figure}
\includegraphics[width=1.0\linewidth,type=pdf,ext=.pdf,read=.pdf]{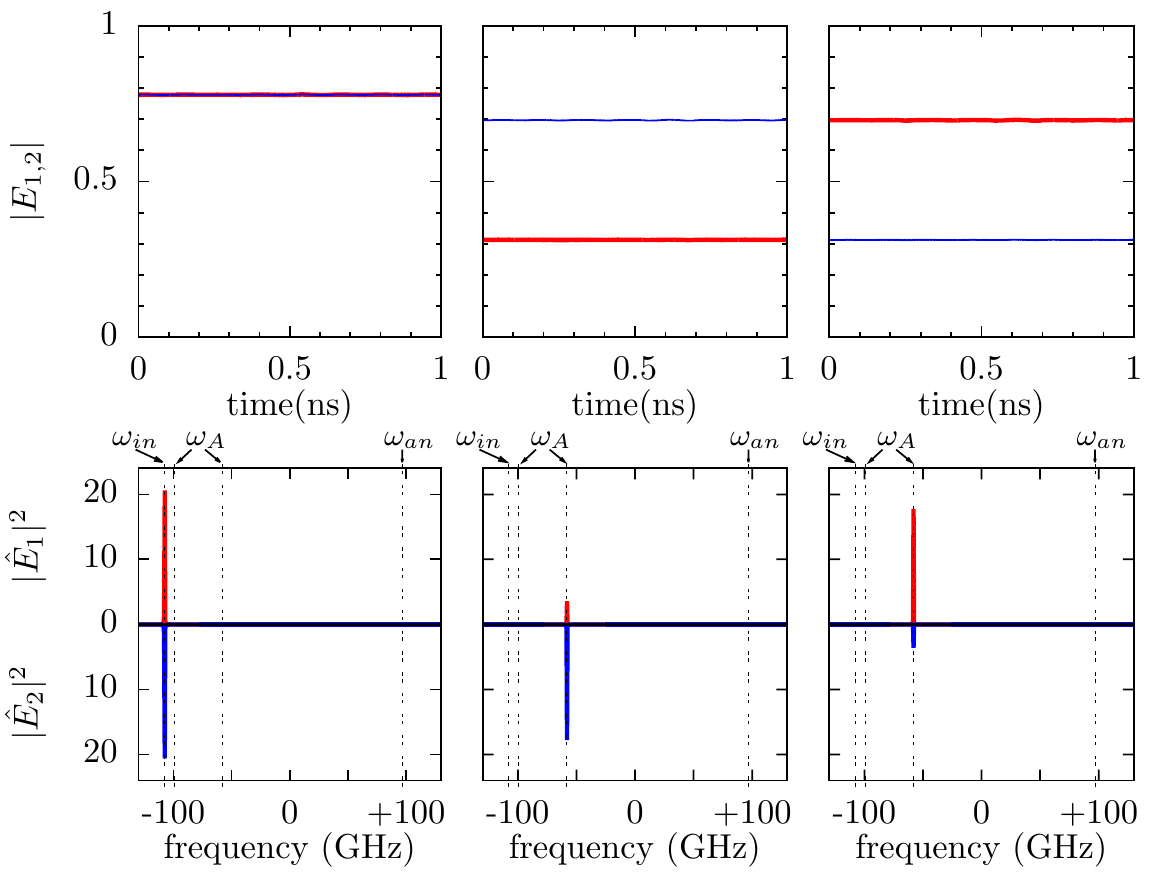}
\caption{\label{fg:CLMs} Time-traces (top panels) and frequency spectra (bottom panels) showing symmetric (left column) and symmetry-broken (middle and right columns) CLMs for $\tau=0.1$, $\kappa=0.3$ and ${C_p}=0.33\,\pi$.  These parameters are consistent with a point in region 7 of Fig.~\ref{fg:map}.}
\end{figure}

One-colour states, which are also known as compound laser modes (CLMs)~\cite{ERZ06} are constant amplitude and single frequency solutions to~\eqref{eq:mLK} whereby the two lasers are frequency locked. They are characterised by the following ansatz,  
\begin{equation}
\label{eq:1C}
\begin{aligned}
& E_1(t) = A_1 e^{i\,\omega_A\,t}& \qquad &N_1(t)= N_1^c\\
& E_2(t) = A_2 e^{i\,\omega_A\,t} e^{i\,\delta_A}& &N_2(t)= N_2^c
\end{aligned}
\end{equation}
with real constants $A_1$, $A_2$, $N_1^c$, $N_2^c$, $\delta_A$ and $\omega_A$. The frequency $\omega_A$ is the common locked  frequency of the slowly varying optical fields of the two lasers and $\delta_A$ allows for a constant phase difference between them. This ansatz can be split into two complementary classes; (i) symmetric CLMs ($A_1=A_2$) which are invariant under the $\mathbb{Z}_2$ symmetry and (ii) symmetry-broken CLMs where both lasers lase at different intensities ($A_1\neq A_2$). Further details on this classification are given in Appendix~\ref{ap:ssbCLM}. 

Symmetric CLMs can be further sub-divided into ``in-phase'' ($\delta_A=0$) and ``anti-phase'' ($\delta_A=\pi$) solutions and their stability is extensively studied in \cite{ERZ06,YAN04}.  In particular, it was found that symmetric CLMs can lose their stability via Hopf, saddle node or pitchfork bifurcations, and the stability boundaries were obtained via numerical continuation techniques \cite{ERZ06} or in the instantaneous limit $\tau=0$ analytically by using the characteristic equation of the system \cite{YAN04}.  In the literature, symmetry-broken CLMs play a role in the stability analysis of symmetric CLMs, but are themselves not stable. 

A principal result of this paper is that \emph{symmetry-broken} CLMs are shown to be stable for small delay and relatively high coupling between the lasers. This is demonstrated numerically in the middle and right columns of Fig.~\ref{fg:CLMs} where optical field intensities and frequency spectra of two stable symmetry-broken states are plotted.  Due to the $\mathbb{Z}_2$ symmetry of exchanging the two lasers, symmetry-broken CLMs always exist in pairs.  For purposes of display, a parameter set was chosen where a symmetric CLM is also stable as shown in the left column of Fig.~\ref{fg:CLMs}, giving rise to a tri-stability in CLMs.  

The frequencies of the CLMs in the bottom panels of Fig.~\ref{fg:CLMs}, can be analytically determined by plugging the ansatz~\eqref{eq:1C} into the system of ODEs~\eqref{eq:mLK} \cite{ERZ06,YAN04}. The equations for $\dot{E}_1$ and  $\dot{E}_2$ yield  the following relation between the frequency $\omega_A$ of the CLM and the phase difference $\delta_A$ between the lasers 
\begin{equation}
\label{eq:freqs}
{\frac{\omega^{2}_A}{ \kappa^{ 2} \left( 1+ \alpha^2 \right)} } = \sin^{2}\left( {C_p} + \omega_A\tau + \arctan\alpha\right) - \sin^{ 2}\delta_A
\end{equation}
For in-phase $\omega_{in}=\omega_A(\delta_A=0)$ and anti-phase $\omega_{an}=\omega_A(\delta_A=\pi)$ CLMs, this immediately gives an implicit function for the frequency.  For the  symmetry-broken one-colour states one additionally needs to consider the fixed point solutions $\dot{N}_{1,2}=0$ for the inversions. An implicit solution is calculated in 
Appendix~\ref{ap:omegaA} for the system with delay time $\tau$. A secant method is then used to obtain the numerical values appearing as dot-dashed vertical lines in Fig.~\ref{fg:CLMs}.

\section{Two-Colour States}\label{sec:2C}

\begin{figure}
\includegraphics[width=1.0\linewidth,type=pdf,ext=.pdf,read=.pdf]{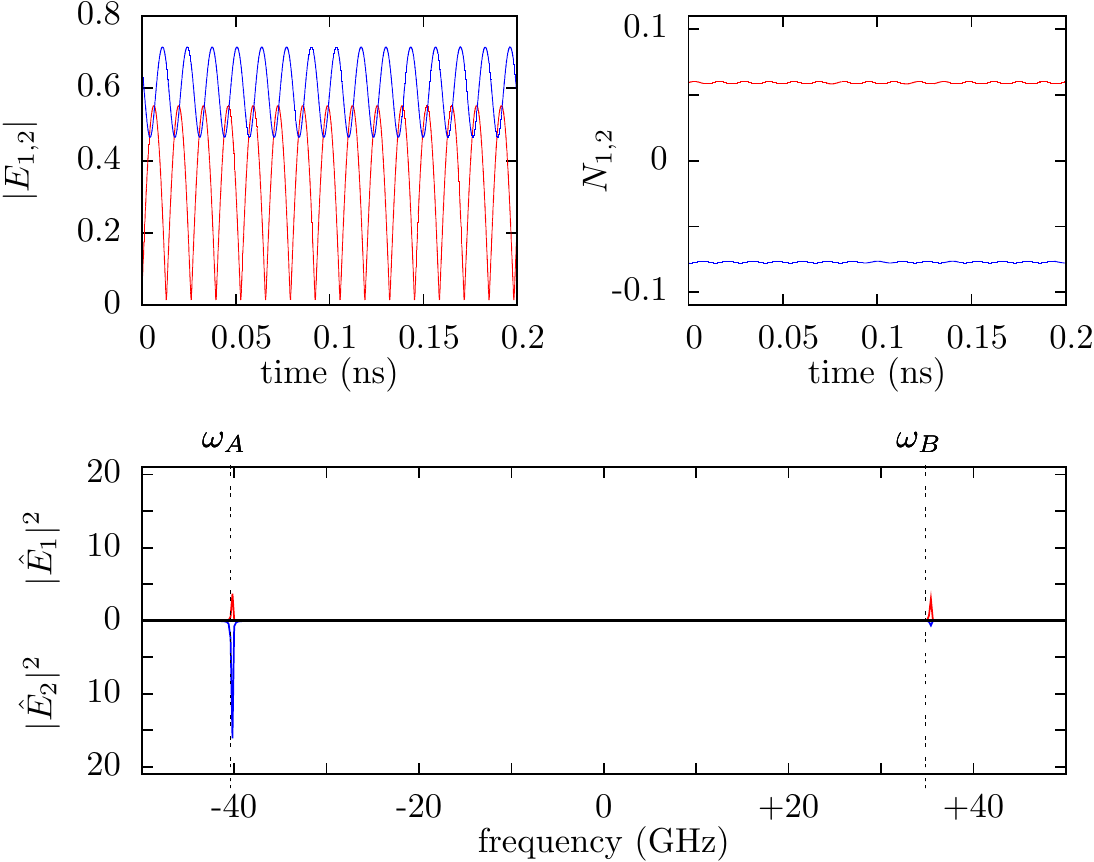}
\caption{\label{fg:2CSB} Numerical solution of system~\eqref{eq:mLK} for $\tau=0.3, \kappa=0.15, C_p=0.33 \pi$ which is consistent with region 4 of Fig.~\ref{fg:map}. The top left panel contains the magnitude of the electric fields, and the top right their inversions. The bottom panel shows the optical spectrum relative to the frequency of the free running lasers. The dot-dashed vertical lines are the frequencies calculated from ansatz~\eqref{eq:2C}.}
\end{figure}

\begin{figure*}[t]
\centering
\includegraphics[width=1.0\linewidth]{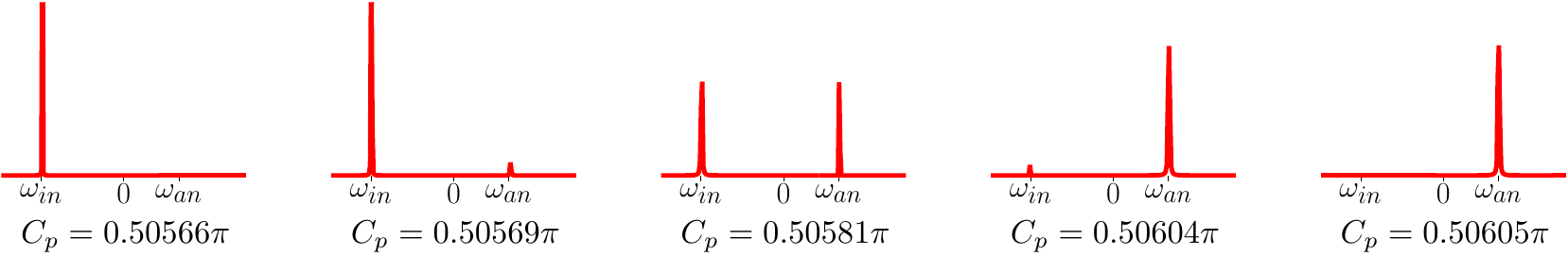}
\caption{\label{fg:BS} Transition between in-phase (left) and anti-phase (right) CLMs for $\tau=0.2$, $\kappa=0.4$.
The in-phase ($\omega_{in} = -74.19~GHz$) and anti-phase ($\omega_{an} = +49.46~GHz$) one-colour frequencies are obtained from Eq.~\eqref{eq:freqs}.}
\end{figure*}

In this section we introduce two-colour states which are stable for small delay and moderate to low coupling strength between the lasers.  We stress that the two-colour states are induced by the coupling between the lasers alone, the uncoupled lasers are single-mode only.  Like the one-colour states, the two-colour states can either be symmetric or symmetry-broken with respect to the $\mathbb{Z}_2$ symmetry of exchanging the two lasers. 

An example of a stable symmetry-broken two-colour state is shown in Fig.~\ref{fg:2CSB}.  Due to the $\field{Z}_2$ symmetry of being able to exchange the two lasers, a twin two-colour state is also stable.  In the optical spectrum in the bottom panel of Fig.~\ref{fg:2CSB} we see that there are indeed only two dominating frequencies $\omega_A$ and $\omega_B$ and both lasers lase at these two frequencies, but with unequal intensities.  In the time traces of the field amplitudes this gives rise to beating oscillations as shown in the upper left panel of Fig.~\ref{fg:2CSB}.   The corresponding inversions shown in the upper right panel of Fig.~\ref{fg:2CSB} also reflects the symmetry-broken nature of this state and in addition shows small oscillations at the beating frequency.

An example for symmetric two-colour states and their connection with symmetric CLMs is shown in Fig.~\ref{fg:BS}.  Starting from the in-phase CLM in the left hand panel of Fig.~\ref{fg:BS} and increasing the coupling phase $C_p$, we observe a torus bifurcation where the frequency of the anti-phase CLM is turned on. Increasing $C_p$ further smoothly transfers power from the in-phase mode to the anti-phase mode until a second torus bifurcation kills the in-phase CLM's frequency and only the anti-phase CLM remains.  Frequency spectrum snapshots as the parameter $C_p$ is changed are shown from left to right in Fig.~\ref{fg:BS}.  We note that symmetric two-colour states exist only over very small ranges of the coupling phase $C_p$, however we will show in Sect.~\ref{sec:map} that they are crucial for the overall understanding of the bifurcation structure in closely coupled single mode lasers. The symmetric two-colour states can be interpreted as beating between symmetric CLMs. This is similar to the beating 
between delay-created external cavity modes for a single laser with mirror \cite{ERN00}.

A useful ansatz~\cite{ROG03} which approximates the dynamics of symmetric and symmetry-broken two-colour states is given by
\begin{equation}
\label{eq:2C}
\begin{aligned}
E_1(t)=A_1 e^{i\,\omega_A\,t}& + B_1 e^{i\,\omega_B\,t},  &N_1(t)= N_1^c,\\
E_2(t)=A_2 e^{i\,\omega_A\,t}& e^{i\,\delta_A} + B_2 e^{i\,\omega_B\,t} e^{i\,\delta_B}, &N_2(t)= N_2^c,
\end{aligned}
\end{equation}
with real constants $A_{1,2}$, $B_{1,2}$, $N_{1,2}^{c}$, $\omega_{A,B}$ and $\delta_{A,B}$.  While this ansatz is a straightforward generalisation of the CLM ansatz \eqref{eq:1C}, with a second frequency $\omega_B$, we stress that in contrast to the CLM ansatz, equations~\eqref{eq:2C} fulfil the original system~\eqref{eq:mLK} only approximately.   The presence of two frequencies give rise to oscillations in the intensities of the electric fields of the form 
\begin{equation}\label{eq:l2norm}
\left\vert E_{1}\left(t\right) \right\vert^2 = A_{1}^2 + B_{1}^2 + 2 A_{1} B_{1} \cos\left(\left(\omega_A -\omega_B\right) t\right)
\end{equation}
and similarly for $\left|E_2\right|^2$.   According to \eqref{eq:mLK} this then also leads to oscillations in the population inversions, as we have seen in the upper left panel of Fig.~\ref{fg:2CSB}, and thereby contradicts the assumptions of constant $N_{1,2}(t)= N_{1,2}^c$.  However, as the parameter $T$ is large, ansatz  \eqref{eq:2C} is often well justified in practice, in particular if the beating frequency $\omega_A - \omega_B$ is also large.  In the same way as for CLMs, the ansatz~\eqref{eq:2C} can be split into symmetric ($A_1=A_2$, $B_1=B_2$) and symmetry-broken solutions. 

In the case of symmetric two-colour states, the frequencies  $\omega_A$ and $\omega_B$ correspond to the frequencies of in-phase and anti-phase CLMs and are obtained from \eqref{eq:freqs}.  For symmetry-broken states the analytical calculation of  $\omega_A$ and $\omega_B$ is shown in Appendix~\ref{ap:omegaA_omegaB}.

\section{Reduced Coordinate System}\label{sec:coord}

Our aim is to understand the bifurcation structure associated with the various one-colour states and two-colour states presented in the previous sections.  As we are mostly interested in the closely coupled limit, we now formally set $\tau=0$.  System \eqref{eq:mLK} then becomes a six dimensional system of ordinary differential equations.  

The system still possesses the  $S^1$ symmetry~\eqref{eq:symm} which can be exploited in order to reduce the number of dimensions of the system.  One popular way of achieving this is by rewriting the system \eqref{eq:mLK} with the dynamical variables $\left(\left| E_{1} \right|,\left| E_{2} \right|,\phi_D, \phi_A, N_1, N_2 \right)$ using $E_{1,2}=\left| E_{1,2} \right| e^{\phi_A\mp \phi_D/_2}$.  Then the dynamical variable for the absolute phase $\phi_A$ decouples from the rest of the system and we are left with a five dimensional system. 

\begin{figure}
\includegraphics[width=0.6\linewidth]{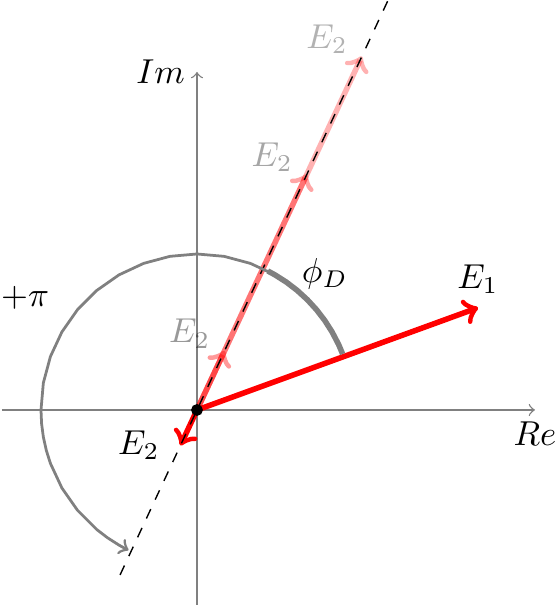}
\caption{\label{fg:singularity} Demonstration of the discontinuity of $\phi_D$ at the origin. A small change in the electric field $E_1$ of the first laser can lead to a large change in the polar coordinate $\phi_D$ to $\phi_D+\pi$.}
\end{figure}

Although widely used, this approach is problematic, because the dynamical variable for the phase difference between the electric fields $\phi_D$ is not well defined if either $E_1$ or $E_2$ vanishes. As a consequence, $\phi_D$  can jump discontinuously as one of the electric fields goes through the origin.  This is schematically demonstrated in Fig.~\ref{fg:singularity}.  To see this mathematically  consider the differential equation for $\phi_{ D}$ which is given by
\begin{equation}
\begin{aligned}
\dot{\phi}_{ D} = &\alpha \left(N_2 - N_1\right)  \\
& + \kappa \left[ \frac{\left|E_2\right|}{ \left|E_1\right| } \sin({ C_p} { -}  \phi_{ D} ) - \frac{\left|E_1\right| }{ \left|E_2\right| } \sin({ C_p} { +}  \phi_{ D} ) \right] \label{eq:phidot} \nonumber     
  \end{aligned}
\end{equation}
The discontinuity in $\phi_D$ manifests itself in the form of singularities at $\left|E_{1/2}\right|=0$, which make it difficult to use numeric continuation software to explore the dynamical features of the system.

In order to avoid these singularities, we introduce a five dimensional coordinate system $(q_x,q_y,q_z,N_1,N_2)$, where the variables are defined via 
\begin{subequations}\label{eq:q}
\begin{eqnarray}
q_x + i q_y &=& 2 E_1^* E_2 \label{eq:qxy}\\
q_z &=& \vert E_1 \vert^2 - \vert E_2 \vert^2 
\end{eqnarray}
\end{subequations}
The multiplication by 2 in \eqref{eq:qxy} ensures that the Euclidean length of the q-vector equals the total intensity output of both lasers, ${R}=\left( q_x^2 + q_y^2 + q_z^2 \right)^\frac{1}{2} = \vert E_1 \vert^2 + \vert E_2 \vert^2$.  The coordinates $q_x, q_y, q_y$ are analogous to the usual Poincar{\'e} sphere representation of polarised light, and per definition do not depend on the absolute phase $\phi_A$.   Therefore the new variables are invariant under the $S^1$ symmetry of the original system \eqref{eq:symm}.  The $\field{Z}_2$ symmetry  now operates as follows
\begin{equation}
\label{eq:symm_new}
 \left(q_x,q_y,q_z,N_1,N_2\right) \rightarrow \left(q_x,-q_y,-q_z,N_2,N_1\right) 
\end{equation}

Rewriting the system \eqref{eq:mLK} for $\tau=0$  in terms of the new coordinates~\eqref{eq:q} yields the five dimensional system
\begin{equation}\label{eq:qdot}
\begin{aligned}
\dot{q_x} &= q_x \left( N_1 + N_2 \right) + \alpha q_y \left( N_1 - N_2 \right) + 2 {\kappa} {R} \cos({C_p}) \\
\dot{q_y} &= q_y \left( N_1 + N_2 \right) - \alpha q_x \left( N_1 - N_2 \right) - 2 {\kappa} q_z \sin({C_p}) \\
\dot{q_z} &= q_z \left( N_1 + N_2 \right) + {R} \left( N_1 - N_2 \right) + 2 {\kappa} q_y \sin(C_p) \\
T \dot{N_1} &= P - N_1 - \left( 1 + 2N_1 \right) \, \left({R + q_z}\right)/2  \\
T \dot{N_2} &= P - N_2 - \left( 1 + 2N_2 \right) \, \left({R - q_z}\right)/2.
\end{aligned}
\end{equation}

As intended, there are now no singularities in the dynamical variables, and the dynamical equations are invariant under the $\field{Z}_2$ symmetry operation \eqref{eq:symm_new}.  It is this reduced system that will be used for bifurcation analysis in the next section.

\section{Bifurcation Diagram}\label{sec:map}

\begin{figure}
\includegraphics[width=\linewidth]{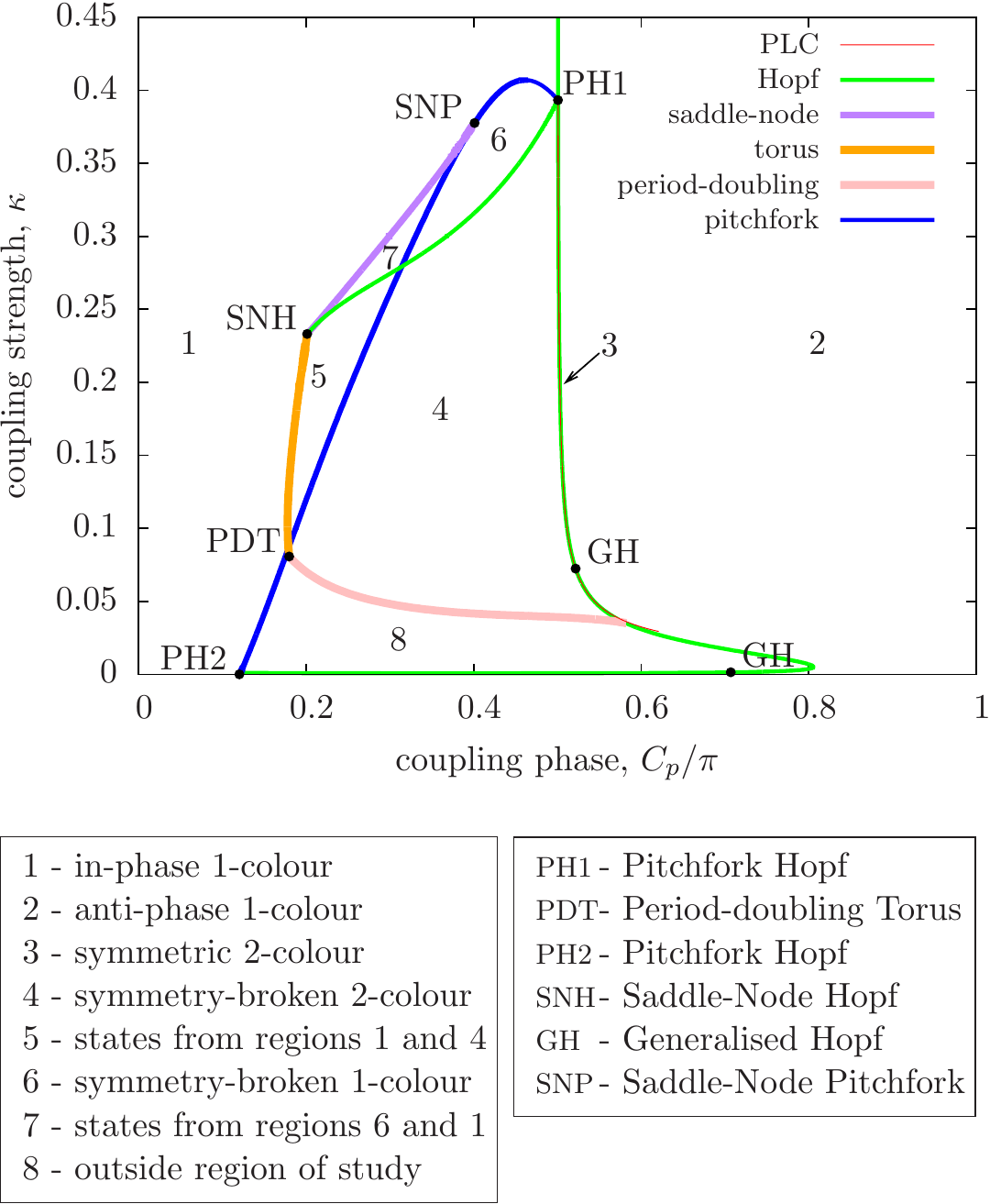}

\caption{\label{fg:map} Bifurcation diagram of system \eqref{eq:qdot} in the $(C_p,\kappa)$ parameter plane. The bifurcation lines separate parameter space into eight distinct dynamical regions.  The codimension one bifurcation lines are organised by a set of codimension two points.  The labelling used for the dynamical regions and for the codimension two points are explained in the left and right legends below the main graph.}
\end{figure}

In this section the bifurcation structure at $\tau =0$ is explored using the 5 dimensional model \eqref{eq:qdot}. In the reduced coordinate system \eqref{eq:q}, one-colour states become fixed point solutions and the two-colour states become limit cycles.  This reduction in complexity allows us use continuation software AUTO \cite{DOE06} to obtain a comprehensive overview of the involved bifurcations. 

In Fig.~\ref{fg:map}, the bifurcation diagram with the two bifurcation parameters of coupling phase $C_p$ and coupling strength $\kappa$ is presented. Due to the parameter symmetry
\begin{equation}\label{eq:psymm}
\left( q_x, q_y, {C_p} \right) \rightarrow \left( -q_x, -q_y, {C_p}+\pi \right) 
\end{equation}
it is sufficient to consider the parameter range of $C_p$ in the interval $[0,1\pi)$ only.  The (codimension one) bifurcation lines in the diagram separate the parameter space into eight distinct regions.  Only bifurcations which affect stable dynamical states are plotted.  In region 1 and 2 in-phase one-colour states and anti-phase one-colour states are stable respectively. For the new coordinates \eqref{eq:q}, symmetric one colour states are confined to  the $q_x$ axis ($q_y=0; q_z=0$), with $q_x>0$ in the in-phase case. 

At very high coupling strengths ($\kappa>PH1$), the in-phase region 1 and the anti-phase region 2 are separated by two supercritical Hopf bifurcations in close proximity as shown in the lower right panel of Fig.~\ref{fg:cuts}. Between these two Hopf bifurcations a stable limit cycle of low amplitude and high frequency exists which corresponds to the symmetric two-colour dynamics seen in Fig.~\ref{fg:BS} of Sect.~\ref{sec:2C}, where the torus bifurcations of the original system \eqref{eq:mLK} have become the Hopf bifurcations in the reduced system. In Fig.~\ref{fg:map} these two vertical Hopf bifurcation lines appear as a single line at the scale of the diagram with the tiny region 3 of symmetric two-colour states nestled between them. In addition symmetry-broken one-colour states ($q_y\neq 0; q_z \neq 0$) are stable in region 6 which is accessible from the symmetric one-colour states via a supercritical bifurcation at the top of the blue pitchfork line between the SNP and PH1 points. 

\begin{figure}
\includegraphics[width=1.0\linewidth]{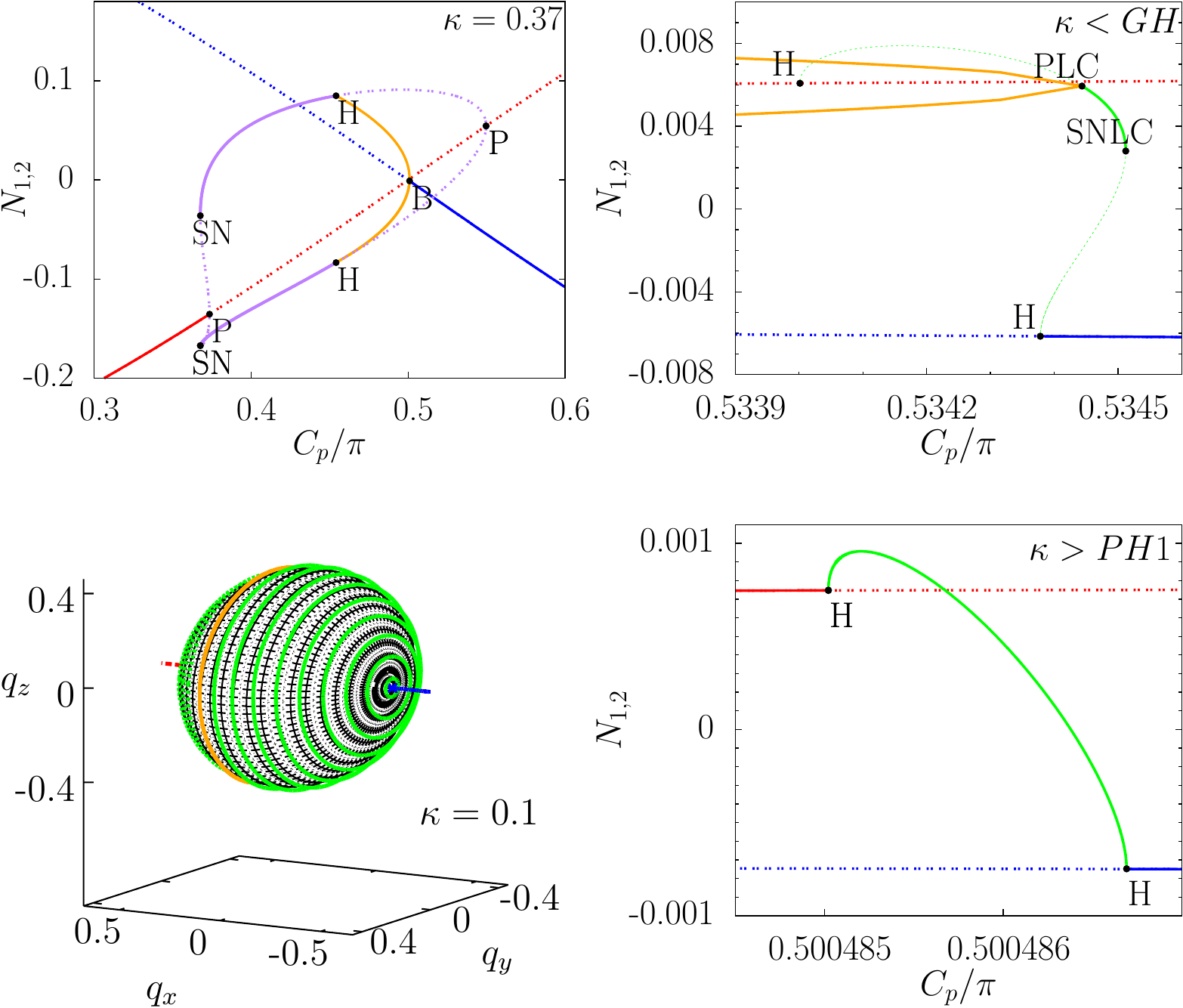}
\caption{\label{fg:cuts} Several cuts of constant $\kappa$ as indicated across Fig.~\ref{fg:map} are shown. The top panels and lower-left panel show bifurcation diagrams of maximum $N_{1,2}$ versus $C_p$. Solid lines indicate stable states and dashed lines unstable states. Symmetric in-phase and anti-phase one-colour states are red and blue respectively. Symmetry-broken one-colour states appear as purple lines whilst symmetry-broken limit cycles are orange. Green lines indicate symmetric limit cycles. Hopf, pitchfork, pitchfork of limit cycles, saddle-node, saddle-node of limit cycles are denoted by H, P, PLC, SN, SNLC. Point B contains several bifurcations, see main text. The lower-left panel shows a 3-dimensional graph in q-vector space~\eqref{eq:q} showing symmetric limit cycles (green) between in-phase (red) and anti-phase (blue) one-colour states. The orange ring marks the pitchfork of limit cycles. The $C_p$ range corresponds to the same as right-most inset of Fig.~\ref{fg:cut}. 
Symmetry-broken 
limit cycles 
not drawn.}
\end{figure}

\begin{figure}
\includegraphics[width=0.9\linewidth]{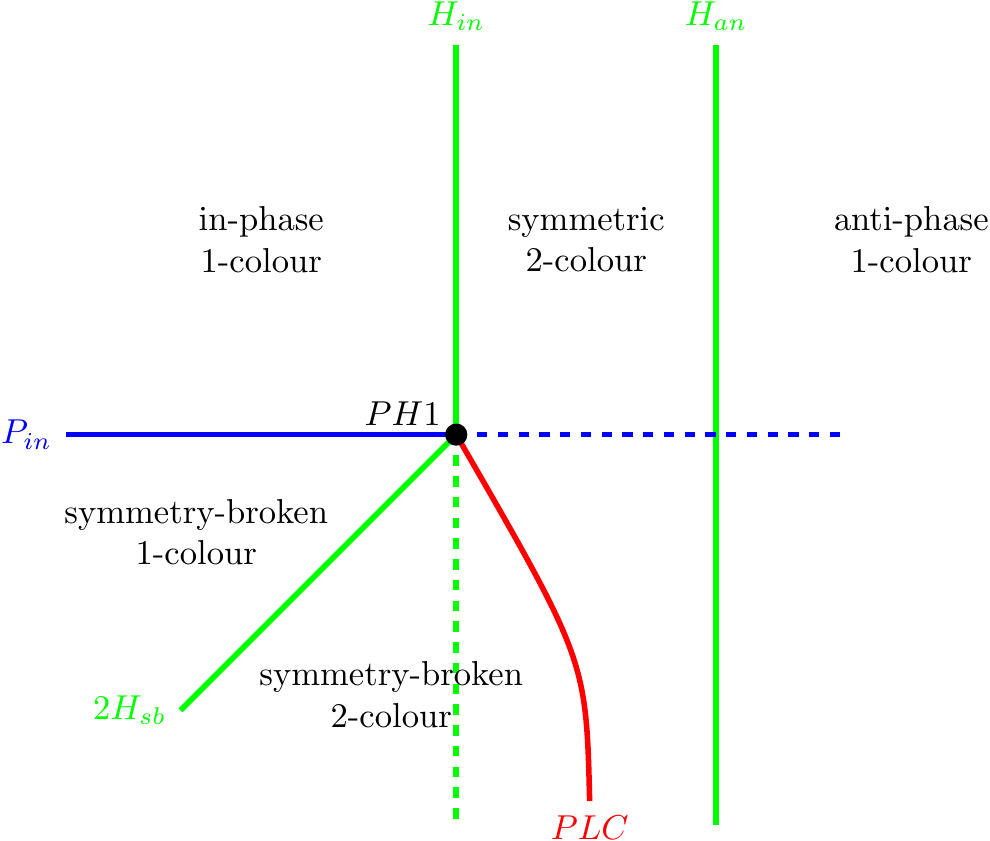}
\caption{\label{fg:blowup} Sketch of bifurcation structure in vicinity of the pitchfork-Hopf codimension two point, labelled PH1 in Fig.~\ref{fg:map}. Subscripts ${in}$, ${an}$ and ${sb}$ refer to bifurcations acting on in-phase, anti-phase and symmetry-broken states respectively.  The solid bifurcation lines affect stable states, while the dashed lines only affect unstable states.  }
\end{figure}

\begin{figure}
\includegraphics[width=1.0\linewidth]{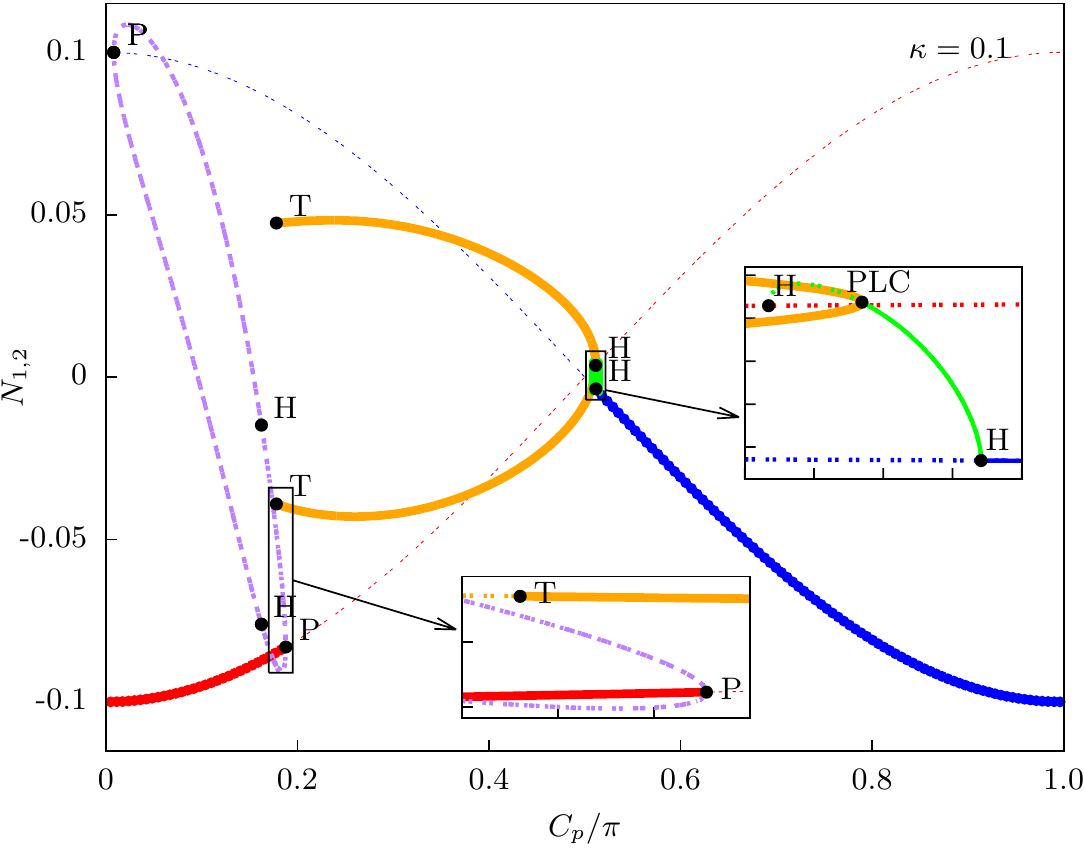}
\caption{\label{fg:cut} An inversion cut at $\kappa=0.1$ across the entire $1\pi$ range of $C_p$. Colours and labels are as in Fig.~\ref{fg:cuts}. T represents a torus bifurcation. Insets are blow-ups of the regions indicated.}
\end{figure}

A sketch of the situation in the vicinity of the PH1 point is provided in the left part of Fig.~\ref{fg:blowup}, where one of these supercritical Hopf bifurcations and the supercritical pitchfork bifurcation intersect. The second supercritical Hopf bifurcation although nearby does not play a role. The PH1 point is a pitchfork Hopf codimension two bifurcation and analytical expressions are obtained in Appendix~\ref{ap:location-ph1-point}. In order to classify this bifurcation, we note that the pitchfork Hopf and the Hopf-Hopf bifurcation share the same reduced normal form. PH1 corresponds to a Hopf-Hopf ``simple'' case III of~\cite{KUZ04} or equivalently case II/III of~\cite{GUC83}. Figure~\ref{fg:blowup} illustrates that this pitchfork Hopf point is responsible for the symmetry-broken two colour dynamics which originates between a pitchfork of limit cycles (PLC) and a supercritical Hopf bifurcation ($2H_{sb}$). Symmetry-broken two-colour states are stable in region 4 of Fig.~\ref{fg:map}. 

The point labelled SNP is a saddle-node pitchfork codimension two bifurcation. It shares the same reduced normal form as a generalised (Bautin) Hopf bifurcation without rotation. At this point the supercritical pitchfork bifurcation becomes subcritical, thus spawning the tri-stable region 7, where both symmetry-broken and symmetric one-colour states are stable. An integration with noise of the original system~\eqref{eq:mLK} showing the three stable one-colour states was presented in Fig.~\ref{fg:CLMs} of section~\ref{sec:CLM}. 

The third corner of symmetry-broken one-colour states is closed by a saddle-node Hopf (also known as fold Hopf) codimension two bifurcation labelled SNH in  Fig.~\ref{fg:map}. At this point a saddle-node and a Hopf bifurcation meet tangentially and also a torus bifurcation line emerges from this point.  The unfolding of saddle-node Hopf bifurcations have been studied extensively in the literature, and we identify the SNH point with the third case in~\cite{KUZ04}\footnote{\cite{KUZ04} does not enumerate the cases. Here we refer to the case of Figure 8.16 in the 3rd version of the book.}. These three codimension two bifurcations organise the multi-stabilities.

Next we provide a series of cuts of constant coupling strength $\kappa$ across the diagram (Fig.\ref{fg:map}) to further explain the dynamics in each of the regions. In Fig.~\ref{fg:cut}, a plot at $\kappa=0.1$ of the inversions $N_{1,2}$ across the entire $1\pi$ span of the coupling phase $C_p$ is presented.  This plot illustrates that symmetry-broken two-colour states are stable over a large range of $C_p$. The bottom inset highlights region 5 which is relatively small at $\kappa=0.1$. In this region a symmetric one-colour is stable in addition to the symmetry-broken two-colour states. In the right-most inset a blow-up of region 3 is provided. The limit cycle born after a supercritical Hopf bifurcation is still invariant under $\field{Z}_2$ symmetry  \eqref{eq:symm_new} and therefore obeys 
\begin{equation}
\label{eq:symm_lc}
 \left(q_x,q_y,q_z,N_1,N_2\right)\left(t\right) = \left(q_x,-q_y,-q_z,N_2,N_1\right)\left(t+\frac{\tau_l}{2}\right),
\end{equation}
where $\tau_l$ is the period of the limit cycle.  In a projection to the $(q_x,q_y,q_z)$ coordinates, the limit cycles therefore form rings around the $q_x$ axis, as shown in the lower left panel of Fig.~\ref{fg:cuts}. These limit cycles corresponds to the symmetric two-colour states. They lose stability at a pitchfork of limit cycles bifurcation and create the large regions 4 and 5 of symmetry-broken two-colour states. In the bifurcation diagram (Fig.~\ref{fg:map}) the GH point is a codimension two generalised Hopf (Bautin) bifurcation. As shown in top right panel of Fig.~\ref{fg:cuts},  the scenario remains very similar for $\kappa$ below this point. However here the supercritical Hopf bifurcation becomes subcritical and a saddle-node of limit cycles bifurcation occurs.  Finally in the top left of Fig.~\ref{fg:cuts}, we look at an inversion cut of much higher $\kappa$ just beneath the SNP 
point. Here the supercritical Hopf and supercritical pitchfork of limit cycles bifurcations occur in close proximity at the point B. In contrast to Fig.~\ref{fg:cut}, symmetry-broken one-colour states are stable limited on the left by a saddle-node bifurcation and on the right by a supercritical Hopf bifurcation. Region 7 of Fig.~\ref{fg:map} corresponds the small area of hysteresis between these symmetry-broken one-colour and symmetric one-colour states. 

The remaining two codimension two points plotted in Fig.~\ref{fg:map} are a period-doubling torus (PDT) and a second pitchfork Hopf (PH2) bifurcation. PH2 gives rise to a tiny triangular region of bi-stability between symmetric one-colour states which occurs below the green horizontal Hopf as shown in~\cite{YAN04}. We identify this point as being equivalent to the Hopf-Hopf ``difficult'' case VI of~\cite{KUZ04} which is the same as case VIa of~\cite{GUC83}. This completes the bifurcation picture for moderate to high coupling strength. We note that different dynamics occurs for lower coupling strengths, in particular region 8 of Fig.~\ref{fg:map} which is outside the scope of the current paper.

\section{All-Optical Switching}\label{sec:opticalswitch}

The realisation of stable, fast and scaleable all optical memory elements would significantly increase the scope of photonic devices and has therefore attracted considerable interest in the laser community. Examples of optical memory designs include hetero-structure photonic crystal lasers~\cite{CHE11}, coupled micro-ring lasers~\cite{HIL04} and dual mode semiconductor lasers with delayed feedback~\cite{BRA12}. In this section, we propose a conceptually {\em simple} all optical memory element on the basis of two closely coupled single-mode lasers.

\begin{figure}
\includegraphics[width=1.0\linewidth,type=pdf,ext=.pdf,read=.pdf]{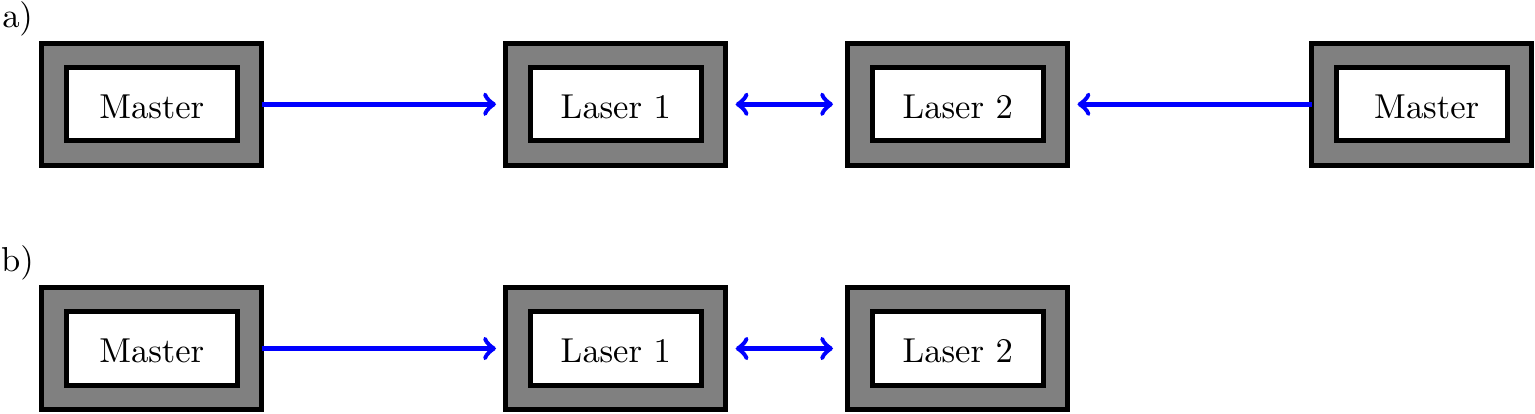}
\caption{\label{fg:scheme2} Sketch showing the two laser coupling schemes considered as optical switches.}
\end{figure}

Memory units require at least two stable states and a mechanism to switch between them. In the previous sections, a number of different parameter regions displaying multi-stabilities were discovered for closely coupled lasers. In particular, in section~\ref{sec:2C} corresponding to region 4 of Fig.~\ref{fg:map} we established the presence of symmetry-broken two-colour states. Switching between them is achieved via optical injection from two master lasers as shown schematically in Fig.~\ref{fg:scheme2}(a). Initially the lasers are in a symmetry-broken state with an optical spectrum displayed in the top left panel of Fig.~\ref{fg:switch_2c}. The oscillations in the magnitude of the electric fields, shown in the graph beneath, are due to the beatings between the two optical frequencies (Sec.~\ref{sec:2C}). Between 5~ns and 6~ns an optical pulse from a master laser is injected into laser 1 which causes laser 1 to lock to the external frequency and for its carrier density to significantly reduce. This is shown in 
the top panel of Fig~\ref{fg:os1speed_N}, where we denote the point ``A'' at approximately $50$~ps  when the carrier density of laser 1 is pushed beneath that of laser 2.  After the pulse is switched off at 6~ns, a transient of about 3~ns is visible in the bottom panel of Fig.~\ref{fg:switch_2c} for the coupled lasers to completely settle to the twin symmetry-broken two-colour state which corresponds to the two lasers having been exchanged (top centre panel of Fig.~\ref{fg:switch_2c}).  The contrast ratio between the two stable states is relatively large.  For example, the intensity of the colour with higher frequency changes by a factor of more than four. 

\begin{figure}
\includegraphics[width=1.0\linewidth,type=pdf,ext=.pdf,read=.pdf]{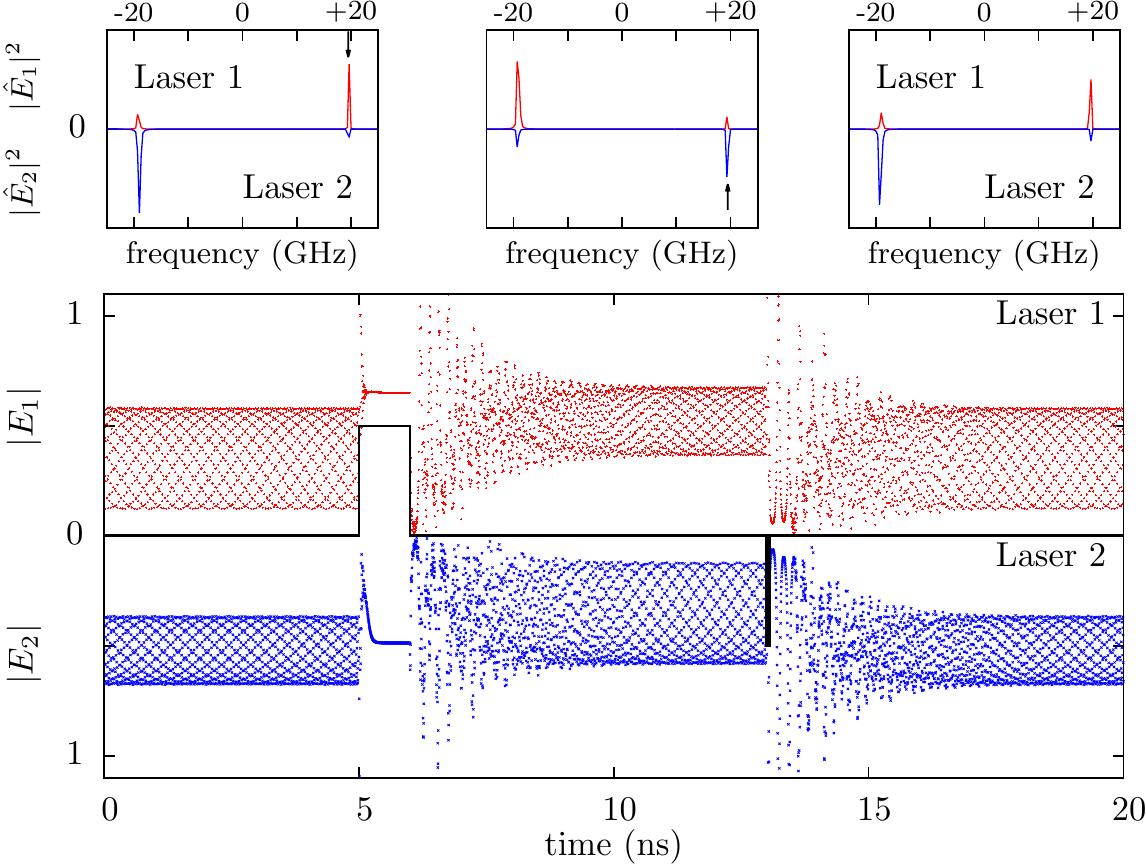}
\caption{\label{fg:switch_2c} Optical switching between symmetry-broken two-colour states for a two laser configuration as outlined in Fig.~\ref{fg:scheme2}(a) for parameters $\tau=0.2$, $\kappa=0.1$, $C_p=0.35\pi$. The large bottom panel shows the magnitude of the electric field in both lasers. Laser 1 in red and laser 2 mirrored underneath in blue. The central black line dividing the two shows the injection strength and duration of the master lasers. At 5~ns a pulse of 1~ns is injected from a master laser into laser 1. At 13~ns a pulse of 50~ps from a master laser is injected into laser 2. The top three panels show the corresponding frequency states. The black arrow indicates the frequency of injected light which is 19.5~GHz larger than the lasers' free running frequency.}
\end{figure}

One important aspect from an application point of view is the ability to switch between two states with a very short external optical pulse and we therefore define the \emph{write time} as the minimum pulse duration to ensure switching.  In order to demonstrate that a shorter pulse is sufficient to trigger the switch, we inject at 13~ns a pulse of 50ps duration into laser 2.  In the lower panel Fig~\ref{fg:os1speed_N}, we see that $N_2$ is indeed pushed below $N_1$ during the pulse. After the pulse is switched off, both  $N_1$ and $N_2$ oscillate strongly but $N_2$ remains consistently below $N_1$. Therefore the write time is determined by the minimum time needed to reduce the carrier density of one laser below that of the other laser which for our setup is approximately 50~ps. We have thus demonstrated a basic mechanism for an all-optical memory element with a large contrast ratio, a short write time and a low coupling strength between the lasers. 

\begin{figure}
\includegraphics[width=1.0\linewidth,type=pdf,ext=.pdf,read=.pdf]{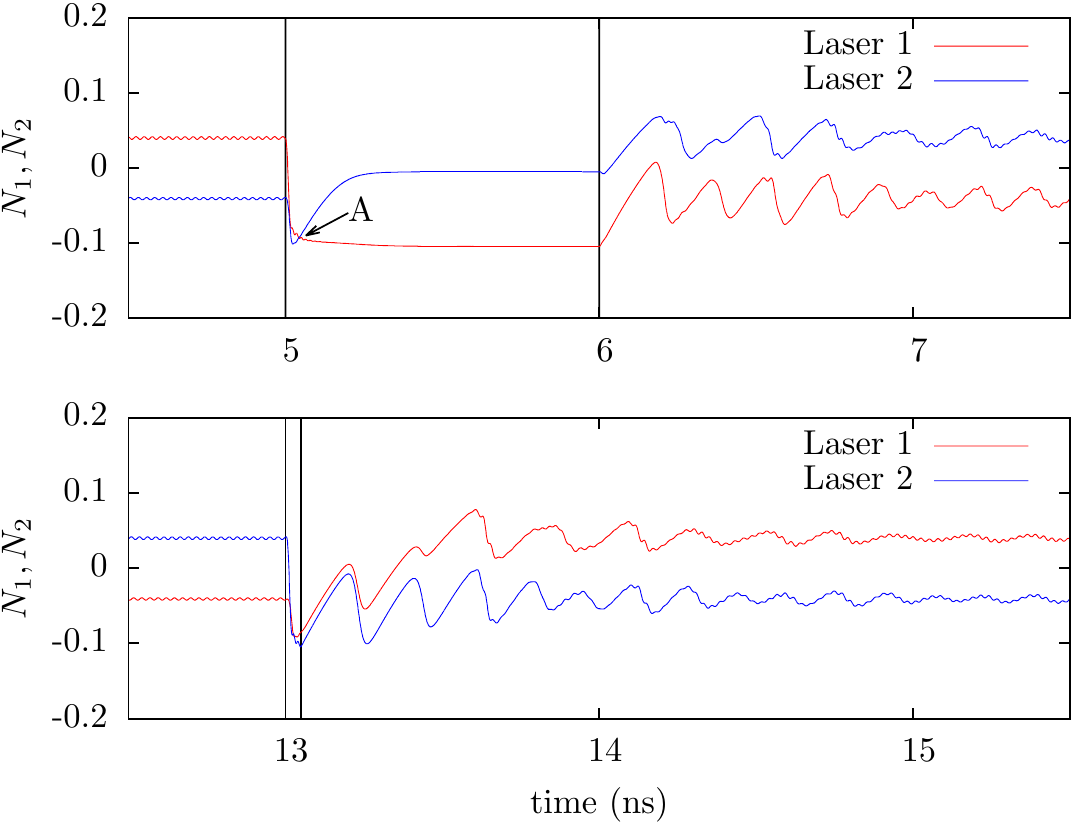}
\caption{\label{fg:os1speed_N} Diagrams showing the dynamics of the carrier densities for both lasers during the first switching event (top panel) and second switching event (bottom panel) of Fig.~\ref{fg:switch_2c}. Black vertical lines mark the time when the external injection is turned on and then off. Point labelled ``A'' indicates where the carrier density of laser~1 becomes less than laser~2.}
\end{figure}

For certain applications it may be desirable to inject into only one of the two lasers as in Fig.~\ref{fg:scheme2}(b). To achieve this, we choose parameters consistent with region 5 of Fig.~\ref{fg:map}. In this region, in addition to the two-colour symmetry-broken states of the previous paragraph, a one-colour symmetric state is also stable. In Fig.~\ref{fg:switch_1c2c}, these states form the basis for the optical memory unit. Initially the two lasers start in a degenerate state with the same amplitude and the same single frequency in both lasers. Using a pulse with a positive detuning relative to the central frequency ($+83$~GHz in the case of Fig.~\ref{fg:switch_1c2c}) the symmetric state of the two lasers symmetry breaks to a two-colour state. As the number, position and intensities of frequencies change, the two states can be easily distinguished which is desirable from an application point of view. 

\begin{figure}
\includegraphics[width=1.0\linewidth,type=pdf,ext=.pdf,read=.pdf]{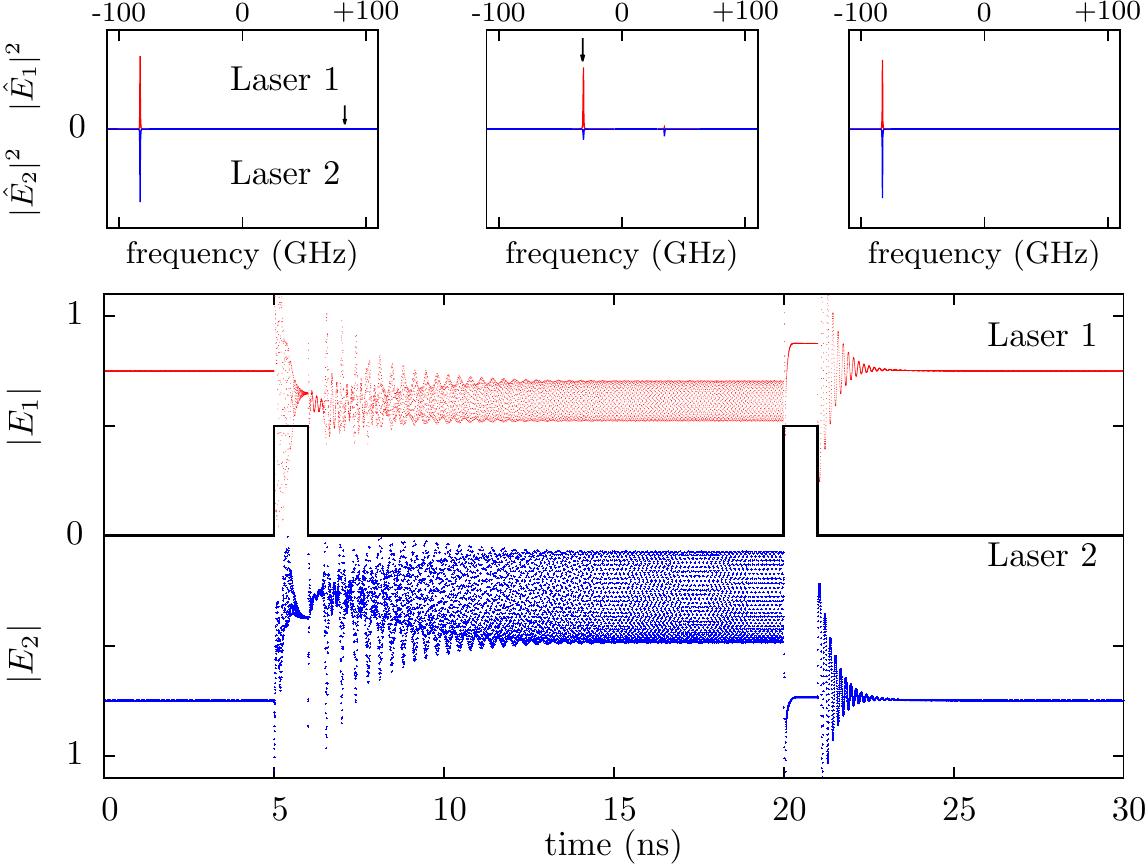}
\caption{\label{fg:switch_1c2c} Optical switching between a symmetric one-colour state and a symmetry-broken two-colour state as sketched in Fig.~\ref{fg:scheme2}(b) for parameters $\tau=0.2$, $\kappa=0.2$, $C_p=0.25\pi$. Panel layout as in Fig.~\ref{fg:switch_2c}. At 5~ns a pulse of 1~ns with frequency offset of +83~GHz is injected from the master laser into laser 1. At 20~ns a pulse of 1~ns with frequency offset of -32~GHz is injected from the master laser into laser 1.}
\end{figure}

To discuss how quickly the two states can be ascertained, we introduce the \emph{read time} which is the minimum duration needed to differentiate optically between the two states after the injection has turned off. In Fig.~\ref{fg:os1speed_N} for the first optical switch after the external injection was removed at 6~ns and 13.05~ns, large amplitude oscillations in the carrier densities at a frequency consistent with the relaxation oscillations for the coupled system (5.7~GHz) ensue. These oscillations are directly related to the length of transient observed in the optical fields for the system to completely settle to the twin symmetry-broken two-colour state. We stress that one does not need to wait for all relaxation oscillations in the system to die out before reading. Small amplitude high frequency oscillations (40~GHz) which are due to the beatings between the two optical colours (see Sec.~\ref{sec:2C}) are observable before each switching event and are discernible in Fig.~\ref{fg:os1speed_N} as small 
deformities in the larger amplitude oscillations within $1$~ns of the external injection being turned off. In Fig.~\ref{fg:os2_freq}, two Fourier modes consistent with the peak frequencies of $-32$~GHz and $-83$~GHz are traced out for each switching event for the second optical switch. All other frequencies are filtered out. During the first injection episode from $5$~ns to $6$~ns, the frequency centred at $-83$~GHz is turned off within $400$~ps. We also observe that after injection, the frequency at $-32$~GHz has reached a stable intensity within 100~ps. To switch back to the symmetric one-colour state, a pulse with negative detuning (-38~GHz in the case of Fig.~\ref{fg:switch_1c2c}) is injected into laser 1. In the lower panel of Fig.~\ref{fg:os2_freq} the frequency component at $-83$~GHz relaxes to its equilibrium value within 1.5~ns. The frequency at $-34$~GHz is completely off within $100$~ps. This constitutes a robust memory element where we inject into one laser only. An external injection pulse with 
positive detuning causes the coupled lasers to enter a symmetry-broken state whilst negative detuning causes the coupled lasers to enter a symmetric state. Compared to the switching scenario in the previous paragraph (Fig.~\ref{fg:switch_2c}) the two lasers are more strongly coupled but the contrast ratio between the two states is greater.

\begin{figure}
\includegraphics[width=1.0\linewidth,type=pdf,ext=.pdf,read=.pdf]{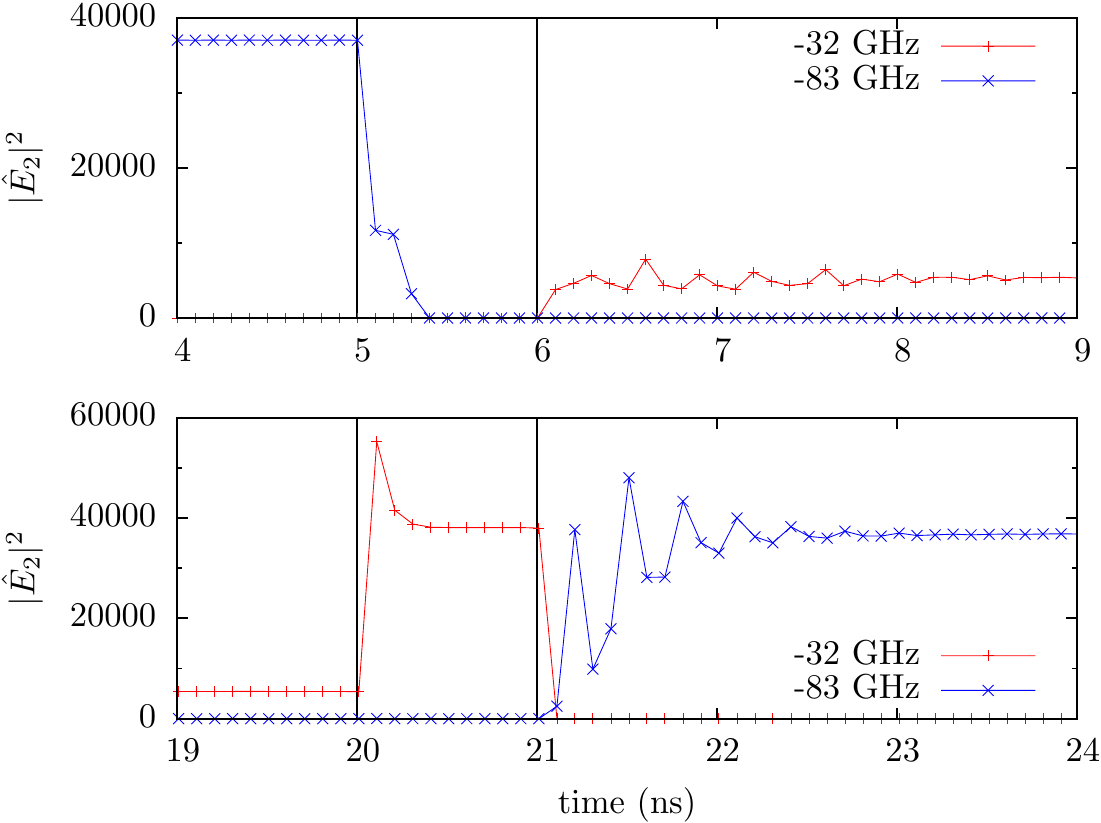}
\caption{\label{fg:os2_freq} Intensity plot tracing frequencies $-32$~GHz and $-83$~GHz for the uninjected laser 2 during the first switching event (top panel) and second switching event (bottom panel) of Fig.~\ref{fg:switch_1c2c}. A Fourier transform using a Hann window is executed every $100$~ps.}
\end{figure}

In this section we have shown that two closely coupled identical lasers can operate as an optical memory element and have provided two specific examples. The high degree of multi-stability discovered in the previous sections enables many other designs. As the multi-stabilities persist in the limit of $\tau\rightarrow0$, the distance between the lasers can be reduced as far as is technologically possible. Memory units of this kind are therefore open to miniaturisation and allow integrability. For the optical switch of Fig.~\ref{fg:switch_2c}, a fast write time of $50$~ps was demonstrated in the second switching event and was discussed with reference to the carrier density in Fig.~\ref{fg:os1speed_N}. A significant speed increase on this number may be possible by increasing the injection strength of the external pulse or a fine tuning of the coupled lasers' parameters. For the optical switch of Fig.~\ref{fg:switch_1c2c}, we showed via Fig.~\ref{fg:os2_freq} that it is possible to distinguish optically between 
the states  within a read time of $100$~ps after the external injection was turned off. Again significant improvements on this number may be possible. Indeed the larger the frequency separation between states the shorter the time needed to differentiate between them. Therefore choosing parameters consistent with larger frequency separations which normally occur at higher coupling strength may substantially decrease the read time. We conclude that closely coupled lasers offer a promising approach for the realisation of scaleable and fast all optical memory elements.



\section{Conclusion}

In this paper, a comprehensive overview of the bifurcation scenarios in a system of two closely coupled single-mode lasers is provided.  For moderate to high coupling strength the four characteristic stable states are symmetric one-colour, symmetry-broken one-colour, symmetric two-colour and symmetry-broken two-colour states.  We introduce a new coordinate representation which accounts for the $S^1$ symmetry in the system without creating unnecessary singularities. This allows us to study the bifurcation structure of this system using conventional numerical continuation techniques.  

Our results show that the bifurcations between the various stable states are organised by a number of codimension two bifurcation points which are identified with reference to the literature.  In particular it is found that the interplay between a pitchfork Hopf, a saddle-node Hopf and a saddle-node pitchfork codimension two points give rise to regions of multi-stabilities.  Detailed knowledge of the bifurcations and their boundaries of closely coupled lasers may open several technological applications. We propose two all-optical switch candidates. 

\begin{acknowledgments}
This work was supported by Science Foundation Ireland under grant number 09/SIRG/I1615. We thank P. Harnedy and E. P. O'Reilly for supporting discussions. 
\end{acknowledgments}

\appendix
\section{Symmetric and Symmetry Broken}\label{ap:ssbCLM}

CLMs or one-colour states are defined in the main text by the ansatz~\eqref{eq:1C} for the delayed coupled system of ODEs~\eqref{eq:mLK}. The following equivalence relation characterise symmetric CLMs whilst symmetric-broken CLMs are the complement class. 

\begin{theorem}\label{thm:sCLM}
The following three conditions are equivalent.
\begin{eqnarray}
(i) \quad &N_1 &= N_2  \nonumber \\
(ii) \quad &A_1 &= A_2 \nonumber \\
(iii) \quad &\delta_A &= 0,\pi  \nonumber
\end{eqnarray}
\end{theorem}

\begin{proof}
Considering equilibrium solutions for $N_{1,2}$ for Eqs.~\eqref{eq:N1dot}~\eqref{eq:N2dot} yields
\begin{eqnarray}\label{eq:AN}
A^2_1 = \frac{P-N_1}{2 N_1 + 1} \qquad A^2_2 = \frac{P-N_2}{2 N_2 + 1} 
\end{eqnarray}
Since $A_{1,2} \ge 0$, functions~\eqref{eq:AN} give 1-1 relation between $N_{1,2}$ and $A_{1,2}$ therefore assuming condition (i) implies (ii). \\

Noting functions~\eqref{eq:AN} are involutions, therefore $N_{1,2}$ can be switched with $A^2_{1,2}$  Then using the same logic as above, assuming (ii) implies (i).

Substituting ansatz~\eqref{eq:1C} into the equations for the electric fields~\eqref{eq:E1dot}~\eqref{eq:E2dot} gives two complex conditions
\begin{eqnarray}
\frac{A_1}{A_2} \left( i\omega_A - (1 + i\alpha) N_1 \right) = \kappa e^{-i{C_p}} e^{i\omega_A \tau} e^{{ +}i\delta_A} \label{eq:E1}\\
\frac{A_2}{A_1} \left( i\omega_A - (1 + i\alpha) N_2 \right) = \kappa e^{-i{C_p}} e^{i\omega_A \tau} e^{{ -}i\delta_A} \label{eq:E2}
\end{eqnarray}
Assuming (ii) then the ratio of the two amplitudes disappear in Eqs.~\eqref{eq:E1}\eqref{eq:E2}. As (ii) implies (i), both left hand sides are equal which gives the condition
\begin{equation}
e^{+i\delta_A} = e^{-i\delta_A} \nonumber  
\end{equation}
Therefore $\sin\delta_A=0$ and $\delta_A=0,\pi$. Therefore assuming (ii) implies (iii).\\

Finally, assuming (iii), i.e $\delta_A=0,\pi$, the right hand sides of Eqs.~\eqref{eq:E1}\eqref{eq:E2} are equal. Combining these two equations and splitting real and imaginary parts gives the following two real conditions;
\begin{eqnarray}
\frac{A_1}{A_2} &=& \frac{N_2}{N_1} \frac{A_2}{A_1} \nonumber \\
\left(\omega_A - \alpha N_1\right) \frac{A_1}{A_2} &=& \left(\omega_A - \alpha N_2\right) \frac{A_2}{A_1} \nonumber
\end{eqnarray}
Substituting one into the other, after cancellations implies $N_1=N_2$ or condition (ii). This suffices as a proof.\qed
\end{proof}

The above conditions justify the distinction between symmetric and symmetry-broken states used in this paper. Furthermore it distinguishes between ``in-phase'' and ``anti-phase'' one-colour symmetric states.

\section{The frequency $\omega_A$ for one-colour states}
\label{ap:omegaA}

In order to obtain the frequency $\omega_A$ appearing in the one-colour ansatz \eqref{eq:1C}, we first derive an explicit relation between the inversion $N^c_1$ and $N^c_2$ as a function of $\omega_A$. Inserting \eqref{eq:1C} into the equations for $\dot{E}_1$ \eqref{eq:E1dot} and $\dot{E}_2$ \eqref{eq:E2dot} yields the matrix equation
\begin{equation}
  \label{eq:A}
  \begin{aligned}
  \left(\begin{array}{cc}
\hat{\alpha} N_{1}^{c}-i\omega_{A} & \kappa e^{-iC_{p}}e^{-i\omega_{A}\tau}\\
\kappa e^{-iC_{p}}e^{-i\omega_{A}\tau} & \hat{\alpha}N_{2}^{c}-i\omega_{A}
\end{array}\right)\left(\begin{array}{c}
A_{1}\\
A_{2}e^{i\delta_{A}}
\end{array}\right)&=0    
\end{aligned}
\end{equation}
where we have used the notation $\hat{\alpha}=1+i\alpha$.  This equation has a non-trivial solutions if its determinant vanishes. For  $\omega_A$ this  leads to the complex condition
\begin{equation}
  \label{eq:2}
  \hat{\alpha}N_{1}^{c}N_{2}^{c}-i\omega_{A}\left(N_{1}^{c}+N_{2}^{c}\right)=\frac{\left[\kappa e^{-iC_{p}}e^{-i\omega_{A}\tau}\right]^{2}+\left(\omega_{A}\right)^{2}}{\hat{\alpha}}
\end{equation}
This allows us to explicitly express $N^c_1$ and $N^c_2$ as functions of  $\omega_A$ via
\begin{equation}
  \label{eq:3}
N_{1/2}^{c}=-\frac{\mbox{Im}\left[\hat{\alpha}^*Q_{A}\right]}{2\omega_{A}}\pm\sqrt{\left(\frac{\mbox{Im}\left[\hat{\alpha}^*Q_{A}\right]}{2\omega_{A}}\right)^{2}-\mbox{Re}\left[Q_{A}\right]}
\end{equation} 
with 
\begin{equation}
  \label{eq:4}
  Q_{A}=\frac{\left[\kappa e^{-iC_{p}}e^{-i\omega_{A}\tau}\right]^{2}+\left(\omega_{A}\right)^{2}}{\hat{\alpha}}.
\end{equation}  Note that the sign $\pm$ in \eqref{eq:3} was chosen in such a way that $N_1^c \geq N_2^c$, however the alternative choice with $\mp$ is equally possible and would imply $N_1^c \leq N_2^c$. 

We can now express $A_2^2 / A_1^2$ in two different ways via,
\begin{equation}\label{eq:A1A2}
{ ( P - N^c_2 )( 2 N^c_1 + 1 ) \over (2 N^c_2 + 1)( P - N^c_1 ) } = { A_2^2 \over A_1^2 } = { (\omega_A - \alpha N^c_1)^2 + (N^c_1)^2 \over \kappa^2 }, 
\end{equation} 
where we have used the two Eqs.~\eqref{eq:AN} on the left hand side,  and the magnitude of the complex expression Eq.~\eqref{eq:E1} on the right hand side.  We then substitute the explicit functions~\eqref{eq:3} of $N_{1/2}^{c}$ into Eq.~\eqref{eq:A1A2} above which gives an implicit solution for the frequency $\omega_A$ of the one-colour state for the system with delay time $\tau$. The frequency $\omega_A$ is finally obtained numerically from this transcendental expression with a secant method yielding the dot-dashed vertical lines in Fig.~\ref{fg:CLMs}.

\section{The frequencies $\omega_A$ and $\omega_B$ for symmetry-broken two-colour states}
\label{ap:omegaA_omegaB}
\begin{figure}
  \centering
  \includegraphics[width=0.7\linewidth,type=pdf,ext=.pdf,read=.pdf]{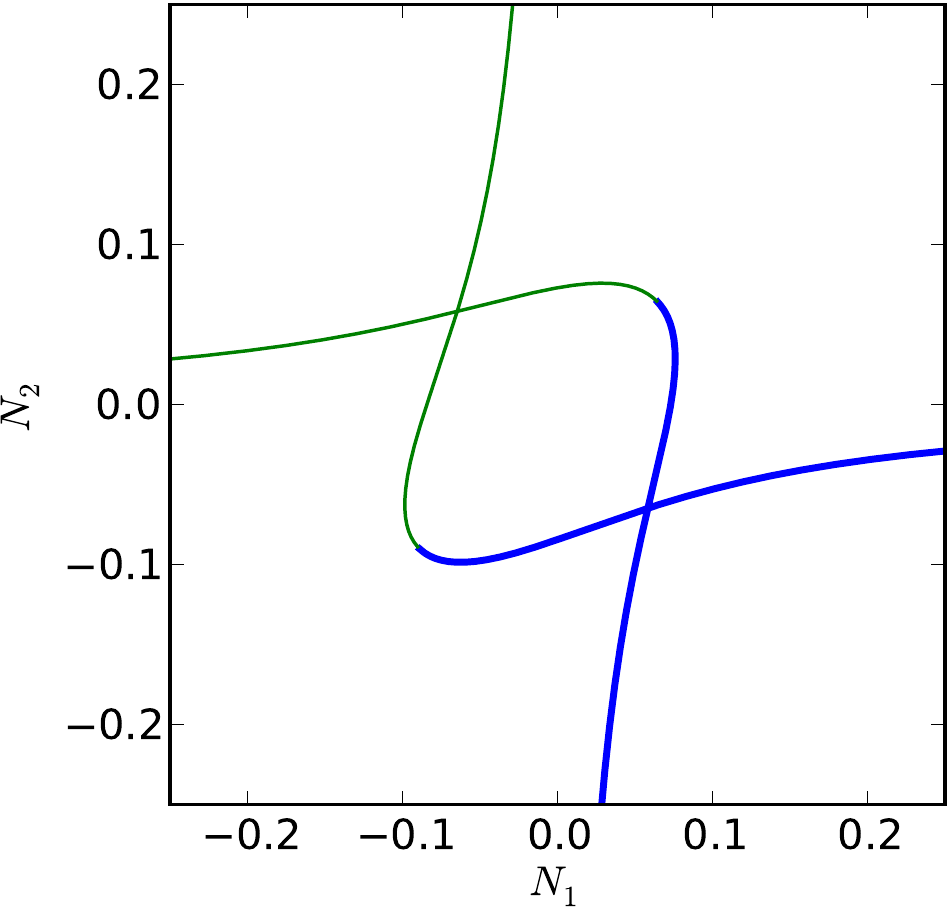}
  \caption{\label{fig:plot_n1_n2} Parametric plot of $\left(N^c_1(\omega_A),N^c_2(\omega_A)\right)$ (blue thick line) according to Eq.~\eqref{eq:3} using the parameter values of Fig.~\ref{fg:2CSB}.  The frequencies $\omega_A$ and $\omega_B$ and the inversions $N^c_1$ and $N^c_2$ of the symmetry-broken two-colour states are obtained from the point where the blue (thick) line self-intersects. The green (thin) line shows a parametric plot of $\left(N^c_2(\omega_A),N^c_1(\omega_A)\right)$. }
\end{figure}
We show, how the frequencies $\omega_A$ and $\omega_B$ appearing in the ansatz \eqref{eq:2C} for the symmetry-broken two-colour states can be obtained analytically. Inserting \eqref{eq:2C} into the equations for $\dot{E}_1$ and $\dot{E}_2$ and sorting the terms for the two frequencies $\omega_{A}$ and  $\omega_{B}$ yields in addition to Eq.~\eqref{eq:A}, a second matrix equation 
\begin{equation}
\label{eq:B}
\begin{aligned}
\left(\begin{array}{cc}
\hat{\alpha} N_{1}^{c}-i\omega_{B} & \kappa e^{-iC_{p}}e^{-i\omega_{B}\tau}\\
\kappa e^{-iC_{p}}e^{-i\omega_{B}\tau} & \hat{\alpha}N_{2}^{c}-i\omega_{B}
\end{array}\right)\left(\begin{array}{c}
B_{1}\\
B_{2}e^{i\delta_{B}}
\end{array}\right)&=0 \\    
  \end{aligned}
\end{equation}

Eqs.~\eqref{eq:A} and \eqref{eq:B} are identical, except for the fact that $\omega_A$ is replaced by $\omega_B$.  The same arguments used in  Appendix~\ref{ap:omegaA} for deriving the expressions for $N^c_{1/2}(\omega_A)$ in Eq.~\eqref{eq:3} also hold for $N^c_{1/2}(\omega_B)$. Therefore explicit expressions for $N^c_{1/2}(\omega_B)$ are obtained by replacing  $\omega_A$  by $\omega_B$ in Eq.~\eqref{eq:3}.  We also have $N^c_1(\omega_A)=N^c_1(\omega_B)$ and $N^c_2(\omega_A)=N^c_2(\omega_B)$.  Therefore we can parametrically plot $\left(N^c_1(\omega),N^c_2(\omega)\right)$ as a function of $\omega$ in the $(N_1, N_2)$ plane as shown in the blue curve of Fig.~\ref{fig:plot_n1_n2} and  the self intersection point of this curve determines the inversions $N^c_1$ and $N^c_2$ and the frequencies $\omega_A$ and $\omega_B$ of the symmetry-broken two-colour states.  We stress that we did not make use of the equations for $\dot{N}_1$ and $\dot{N}_2$ of system \eqref{eq:mLK}, and because the ansatz \eqref{eq:2C} is only 
approximately true, the obtained frequencies are not exact. 

\section{Pitchfork Hopf codimension two points}
\label{ap:location-ph1-point}

The codimension two point PH1 in Figs.~\ref{fg:map} and \ref{fg:blowup} plays an important role in organising the bifurcation structure of system~\eqref{eq:qdot}.  Pitchfork and symmetric Hopf bifurcation lines are obtained in~\cite{YAN04}, here we derive analytical expressions for their intersection and thus the location of the two pitchfork Hopf points, PH1 and PH2 in the $(C_p,\kappa)$ parameter plane.

Let us first evaluate the Jacobian of \eqref{eq:qdot} at an in-phase one-colour state with $q_x=(2 P + \kappa \cos{(C_p)})/(1 - 2\kappa\cos{(C_p)})>0$, $q_y=q_y=0$ and  $N_1 = N_2=-\kappa\cos C_{p}$. In the slightly transformed coordinates $(q_x,N_S,q_y,q_z,N_D)$, where $N_S=(N_1+N_2)/2$ and $N_D=(N_2-N_1)/2$, the Jacobian $J$  has then the form
\begin{equation}
  \label{eq:5}
  J = \left(\begin{array}{cc} 
      J_1 & 0\\ 
      0 & J_2 
    \end{array}\right).
\end{equation}
The matrices $J_1$ and $J_2$ are given by
\begin{eqnarray}
  \label{eq:6}
  J_{1}&=& \left(\begin{array}{cc}
      0 & 2q_{x}\\
      -\frac{1}{2T}\left(1-2\kappa\cos C_{p}\right) & -\frac{1+q_{x}}{T}
    \end{array}\right),\\
  J_{2}&=& \left(\begin{array}{ccc}
      -2\kappa\cos C_{p} & -2\kappa\sin C_{p} & 2\alpha q_{x}\\
2\kappa\sin C_{p} & -2\kappa\cos C_{p} & -2q_{x}\\
0 & \frac{1}{2T}\left(1-2\kappa\cos C_{p}\right) & -\frac{1+q_{x}}{T}
\end{array}\right).
\end{eqnarray}
The Jacobian $J$ therefore decomposes the phase space into two subspaces spanned by the coordinates $(q_x,N_S)$ and $(q_y,q_z,N_D)$ respectively.  At the pitchfork Hopf bifurcation it is required that $J$  has one eigenvalue at zero and an additional pair of complex conjugate eigenvalues on the imaginary axis.  However, we observe that for $\kappa<1/2$ both eigenvalues of $J_1$ have  negative real part.  Therefore it follows that at the pitchfork Hopf bifurcation all three eigenvalues of $J_2$ have zero real part. This in particular implies that  $\text{tr} J_2=0$ which for $0<\kappa<\frac{1}{2}$  gives rise to the condition
\begin{equation}
  \label{eq:7}
  \kappa\cos C_{p}=\frac{1}{4}\left(1-\sqrt{1+\frac{4P+2}{T}}\right)\approx-\frac{2P+1}{4T},
\end{equation}
where the last approximation is correct in first order of $1/T$. We can then use the condition $\det J_2=0$ to obtain an expression for the pitchfork line in Fig.~\ref{fg:map} as 
\begin{equation}
\kappa ( 2 P + 1 ) =  ( \alpha\sin{C_p} - \cos{C_p} ) (1 - 2 \kappa \cos{C_p} ) ( P  + \kappa \cos{C_p} ). 
\end{equation}
Eliminating $C_p$ in the above equation using \eqref{eq:7} yields a quadratic equation in $\kappa^2$ which gives two solutions for positive $\kappa$. For the first non-vanishing order in $1/T$ we find the following expressions for the location of the points PH1 and PH2 in the $(C_p,\kappa)$ parameter plane,
\begin{eqnarray}
\left( C_p, \kappa \right)_{\text{PH1}} &\approx& \left( \cos^{-1}\left[\frac{-\alpha P}{4T} \right], \frac{\alpha P}{2P+1}\right) \nonumber \\
\left( C_p, \kappa \right)_{\text{PH2}} &\approx& \left(-\cos^{-1}\left[\frac{-\alpha}{\sqrt{1+\alpha^{2}}} \right] \nonumber,
\frac{\sqrt{1+\alpha^{2}}}{\alpha} \frac{2P+1}{4T} \right)
\end{eqnarray}
Above expressions are valid for the in-phase one-colour states. The co-dimension two points PH1 and PH2 seen in Fig.~\ref{fg:map} correspond to the in-phase PH1 and the anti-phase PH2. To obtain the corresponding expressions for anti-phase states, the phases $C_p$ can be shifted by $\pi$ in accordance to parameter symmetry~\eqref{eq:psymm}. 


\end{document}